\documentclass[12pt]{article}

\usepackage{fontenc}
\usepackage{color}
\usepackage{float}

\usepackage{amssymb,amsmath,amsthm,amsfonts}
\usepackage{enumitem}
\usepackage{mathtools}
\usepackage{pifont}
\usepackage{units}
\usepackage{verbatim}
\usepackage[svgnames,x11names]{xcolor}
\usepackage{tikz}
\usepackage[utf8]{inputenc}
\usepackage[noabbrev,capitalize,nameinlink]{cleveref}
\usepackage[ruled]{algorithm2e}
\usepackage{mdframed}
\usepackage{kvoptions-patch}

\Crefname{fig}{Figure}{Figures}

\usepackage{caption}
\usepackage{subcaption}
\usepackage{mathabx}
\usepackage[svgnames]{xcolor}

\usepackage{lineno}

\graphicspath{ {img/} }

\newcommand{\group}{$5$-}
\newcommand{\nh}{nh}
\newcommand{\T}{\mathcal{T}}
\newcommand{\com}[1]{\ignorespaces} 
\newcommand{\etal}{{~et~al.}}

\newcommand\tabb[1][0.4cm]{\hspace*{#1}}

\newcommand{\deleted}[1]{}

\newtheorem{theorem}{Theorem}
\newtheorem*{lemma*}{Lemma}
\newtheorem{lemma}{Lemma}
\newtheorem{observation}{Observation}
\theoremstyle{definition}
\newtheorem{definition}{Definition}

\crefname{algo}{Algorithm}{algorithms}
\allowdisplaybreaks

\title{On selecting a fraction of leaves with disjoint neighborhoods in a plane tree\thanks{This work  was supported in part
by the Swiss National Science Foundation, project SNF 200021E-154387. It appears in \emph{Discrete Applied Mathematics} \cite{junginger2021combDAM}.}}

\author{Kolja Junginger, Ioannis Mantas, and Evanthia Papadopoulou\thanks{Corresponding author: evanthia.papadopoulou@usi.ch}}
\date{Faculty of Informatics, USI Università della Svizzera italiana, Lugano, Switzerland}

\begin{document}

\maketitle

\begin{abstract}
We present a generalization of a combinatorial result by Aggarwal,
Guibas, Saxe and Shor~[\textit{Discrete \& Computational Geometry},
1989] on a linear-time algorithm 
that selects a constant fraction of leaves,
with pairwise disjoint neighborhoods,
from a binary tree embedded in the plane.
This result of Aggarwal\etal{} is essential to
the linear-time framework, which they also introduced, that computes certain Voronoi
diagrams of points with a tree structure in linear time. 
An example is the diagram computed while updating the Voronoi diagram of points after
deletion of one site. 
Our generalization allows that only a fraction of
the tree leaves is considered,
and it is motivated by linear-time Voronoi
constructions for non-point sites.
We are given a plane tree $T$ of $n$
leaves, $m$ of which have been marked, and each marked leaf is associated
with a \emph{neighborhood} (a subtree of $T$) such that any two topologically
consecutive marked leaves have disjoint neighborhoods. 
We show how to select in linear time a constant fraction of the
marked leaves having pairwise disjoint neighborhoods. 
\end{abstract}

\section{Introduction}
\nocite{junginger2021combDAM}
In 1987, Aggarwal, Guibas, Saxe and Shor~\cite{AGSS89} introduced a
linear-time technique to compute the Voronoi diagram of points in
convex position, which can also be used to compute other
Voronoi diagrams of point-sites with a tree structure such as:
(1) updating a nearest-neighbor Voronoi diagram of points after deletion of one site;
(2) computing the farthest-point Voronoi diagram, after the convex hull of the points is known;
(3) computing an order-$k$ Voronoi diagram of points, given its order-($k\textup{-}1$) counterpart.
Since then, this framework has been used (and extended) in various
ways to tackle various linear-time Voronoi constructions, including the \emph{medial axis} of a simple polygon by
Chin\etal~\cite{snoyeink99}, 
the \emph{Hamiltonian abstract Voronoi diagram} by Klein and
Lingas~\cite{klein1994}, and 
some \emph{forest-like}
abstract Voronoi diagrams by Bohler\etal~\cite{bohler2018}.
The linear-time construction for problem~(3) improves by a logarithmic
factor the standard iterative construction by Lee~\cite{lee1982k} to compute
the order-k Voronoi diagram of
point-sites, which is in turn 
used in different scenarios; for example, algorithms for  \emph{coverage problems} in wireless
networks by So and Ye~\cite{so2005}.
A much simpler randomized linear-time approach for problems~(1)-(3)
was introduced by Chew~\cite{Chew90}.

\begin{figure}
	\centering
	\centering
	\includegraphics[width=.72\textwidth,page=1]{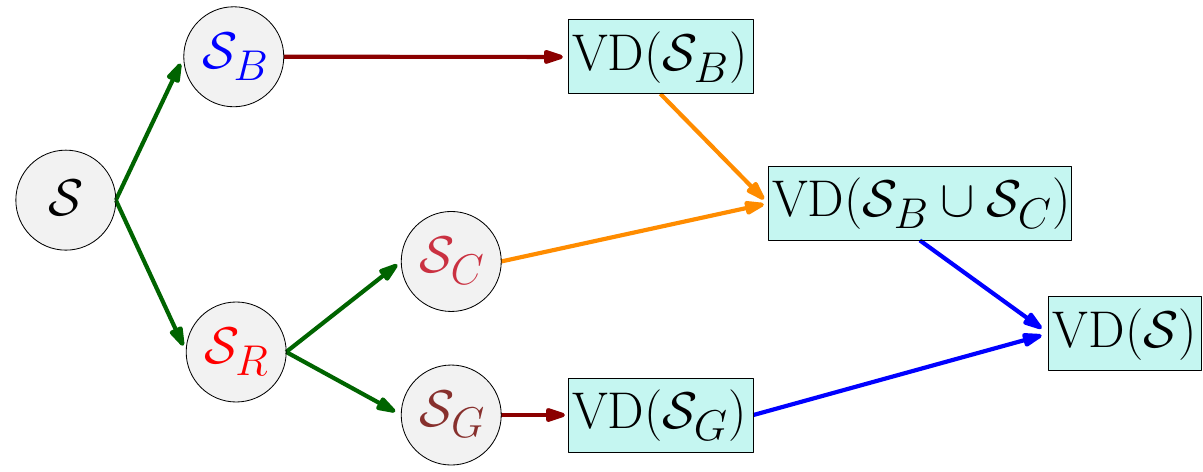}
	\caption{
  The divide \& conquer algorithm of \cite{AGSS89}.
		({\color{DarkGreen}\normalsize{$\longrightarrow$}}) indicates the two different divide phases;
		({\color{DarkRed}\normalsize{$\longrightarrow$}}) indicates the recursive constructions;
		({\color{DarkOrange}\normalsize{$\longrightarrow$}}) indicates site insertion; 
		({\color{DarkBlue}\normalsize{$\longrightarrow$}}) indicates merging. 
	}
	\label{fig:generalRecursion}
\end{figure} 

The linear-time technique of Aggarwal\etal~\cite{AGSS89} is a doubly-recursive
divide-and-conquer scheme
operating on an ordered set of points $S$ whose Voronoi diagram is a
tree with connected Voronoi regions.
At a high level it can be described as follows,
see \cref{fig:generalRecursion}.
In an initial divide phase, the set $S$ is split in 
two sets $\mathcal{S}_R$ (red) and $\mathcal{S}_B$ (blue) of roughly equal size, with
the property that
every two consecutive red sites in  $\mathcal{S}_R$
have disjoint Voronoi regions.
In a second  divide phase, the set $\mathcal{S}_R$ is  split further in sets $\mathcal{S}_C$
(crimson) and $\mathcal{S}_G$ (garnet), 
so that any two sites in $\mathcal{S}_C$ have pairwise disjoint
regions in the Voronoi diagram of $\mathcal{S}_B\cup \mathcal{S}_C$, 
and the cardinality of $\mathcal{S}_C$ is a constant fraction of the
cardinality of $\mathcal{S}_R$.
In the merge phase,  the sites of $\mathcal{S}_C$ are inserted one by one in the recursively
computed Voronoi diagram of $\mathcal{S}_B$, deriving the Voronoi
diagram of $\mathcal{S}_C\cup \mathcal{S}_B$, and the result is  merged with the
recursively computed diagram of $\mathcal{S}_G$.

The key factor in obtaining the  linear-time complexity is that the
cardinality of the set
$\mathcal{S}_C$
is a constant fraction of $\mathcal{S}_R$, which is $\Theta(\mathcal{|S|})$, and
that $\mathcal{S}_C$ can be obtained in linear time.
This is  
possible due to the following combinatorial result of \cite{AGSS89} on a
geometric binary tree embedded in the plane.
This result is, thus, inherently used by 
any algorithm that is based on 
the linear-time framework of Aggarwal\etal{}
A binary tree that contains no nodes of
degree 2 is called \emph{proper}.

\begin{theorem}[\cite{AGSS89}]
	\label{thrm:Aggarwal}
	Let $\T$ be an unrooted (proper) binary tree embedded in the plane.
	Each leaf of $\T$ is associated with a neighborhood, which is
	a (proper) subtree of $\T$ rooted at that leaf;  consecutive leaves in
	the topological ordering of $\T$ have disjoint neighborhoods.
	Then, there exists a fixed fraction of the leaves whose
	neighborhoods are pairwise disjoint, they have a constant size, and no tree edge has its endpoints in two different neighborhoods. 
	Such a set of leaves can be found in linear time. 
\end{theorem}

Overall, the time complexity of the algorithm is described by the following recursive equation and can be proved to be $\Theta(n)$, where $n = |\mathcal{S}|$.
\begin{align*}
	T(n) =& \ T(|\mathcal{S}_B|) + T(|\mathcal{S}_G|) + \Theta(|\mathcal{S}_R|) + |\mathcal{S}_C|\cdot \Theta(1) + \Theta(n)\\
	=& \  T(|\mathcal{S}_B|) + T(|\mathcal{S}_G|) + \Theta(n) \label{eq:combinatorial}\\
	=& \ \Theta(n) \tag{Because $|\mathcal{S}_C| = \Theta(n)$}
\end{align*}

It is worth understanding what \cref{thrm:Aggarwal} represents, in order to have a spherical perspective of its connection to Voronoi diagrams.
An embedded tree corresponds to the graph structure of a Voronoi diagram, and leaves are the endpoints of unbounded Voronoi edges \emph{"at infinity"}; see \cref{fig:combOriginal1}.
The neighborhood of a leaf corresponds to the part of the diagram (of $\mathcal{S}_B$) that gets deleted if a point-site is inserted there; see  \cref{fig:combOriginal2}.
Hence, \cref{thrm:Aggarwal} aims to select leaves with pairwise disjoint neighborhoods ($\mathcal{S}_C$), as they can easily, and independently from one another, be inserted in the diagram. 

\begin{figure}
	\centering
	\begin{subfigure}[t]{0.315\textwidth}
		\centering
		\includegraphics[width=\textwidth,page=1]{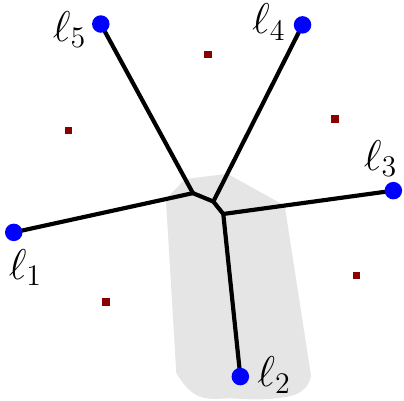}
		\caption{
   The Voronoi diagram of 5 points ({\large\color{DarkRed}$\sqbullet$});
		the neighborhood of $\ell_2$ is shaded. 
		}
		\label{fig:combOriginal1}
	\end{subfigure}
		\hfill	
	\begin{subfigure}[t]{0.315\textwidth}
		\centering
		\includegraphics[width=\textwidth,page=2]{combIntroNew}
		\caption{
   The dashed part of the diagram will get deleted if point ({\large\color{SeaGreen}$\sqbullet$}) is inserted.
		}
		\label{fig:combOriginal2}
	\end{subfigure}	
	\hfill
	\begin{subfigure}[t]{0.315\textwidth}
		\centering
		\includegraphics[width=\textwidth,page=3]{combIntroNew}
		\caption{
			The topological ordering of the leaves in $\T$.
		}
		\label{fig:combOriginal3}
	\end{subfigure}	
	\caption{
		An embedded binary tree $\T$ in the setting of Aggarwal\etal~\cite{AGSS89}.
	}
	\label{fig:combOriginal}
\end{figure} 

For generalized sites,  other than points in the plane, or for abstract
Voronoi diagrams, deterministic linear-time algorithms for the counterparts of
problems (1)-(3) have not been known so far. 
This includes the diagrams of very simple geometric sites such as line segments and circles in the
Euclidean plane.
A major complication over points is that the underlying diagrams  have disconnected
Voronoi regions.
Recently Papadopoulou\etal~\cite{junginger2018,khramtcova2017} presented  a randomized  linear-time
technique for these problems, based on a relaxed Voronoi
structure, called a \emph{Voronoi-like diagram}~\cite{junginger2018,khramtcova2017}.
Whether this structure can be used within 
the framework of Aggarwal\etal, leading  to deterministic linear-time
constructions, remains still an open problem.
Towards resolving this problem we  need a generalized version of
Theorem~\ref{thrm:Aggarwal}. 
   
The problem is formulated as follows. 
We have an unrooted binary tree $\T$ embedded in
the plane, which  corresponds to a Voronoi-like structure.
Not  all leaves of $\T$ are eligible 
for inclusion in the set ${S_C}$  of the linear-time framework. 
As in the original problem, each of the eligible leaves is associated with a neighborhood,
which is a subtree of $\T$ rooted at that leaf, and adjacent
leaves in the topological ordering of $\T$ have disjoint
neighborhoods.
In linear time, we need to compute a constant fraction
of the eligible leaves such that their neighborhoods are
pairwise disjoint.
The non-eligible leaves spread arbitrarily along
the topological ordering of the tree leaves.
This paper addresses this problem by proving the following
generalization of \cref{thrm:Aggarwal}. 

\begin{theorem} 
	\label{thrm:main}
	Let $\T$ be an unrooted (proper) binary tree embedded in the plane
	having  $n$ leaves, $m$ of which have been marked.
	Each marked leaf of $\T$ is associated with a neighborhood, which is a proper subtree of $\T$ rooted at this leaf, 
	and any two 
	consecutive marked leaves in the topological ordering of $\T$ have disjoint neighborhoods.
	Then, there exist at least $\frac{1}{10}m$ marked leaves
	whose neighborhoods are pairwise disjoint
	and no tree
	edge has its endpoints in two 
	of these neighborhoods.         
	Further, we can select at least a fraction $p$ of these $\frac{1}{10}m$ 
	marked leaves in time $O(\frac{1}{1-p}n)$, for any $p \in (0,1)$.
\end{theorem}

The algorithm of Theorem~\ref{thrm:main}
allows for a trade-off between the number of
the returned marked leaves and its time complexity, 
using a parameter $p\in (0,1)$.
If $p$ is constant then the algorithm returns 
a constant fraction of the marked leaves  in $O(n)$ time.
\cref{thrm:main} is a combinatorial result on an embedded tree, and thus, 
we expect it to find applications in different contexts as well.

\begin{figure}
	\centering
	\begin{subfigure}[t]{0.48\textwidth}
		\centering
		\includegraphics[width=\textwidth,page=1]{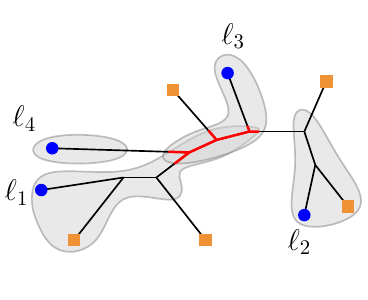}
		\caption{
			Neighborhoods of the marked leaves are shown shaded.
			The intersection of neighborhoods is highlighted with red.
	}
		\label{fig:combMarkedSetting}	
	\end{subfigure}
	\hfill
	\begin{subfigure}[t]{0.48\textwidth}
		\centering
		\includegraphics[width=\textwidth,page=3]{combTransform}
		\caption{
			A marked tree $\T$, which serves as an example instance for illustrating the notions of \cref{sec:prelims}.
		}
		\label{fig:combToriginal}
	\end{subfigure}
	\caption{
		Two marked trees, where marked leaves are shown with ({\large\color{Blue}$\bullet$}) and unmarked leaves are shown with ({\large\color{DarkOrange}$\sqbullet$}).
		}
\end{figure} 

\section{Preliminaries}\label{sec:prelims}
Throughout this work, we consider  an unrooted binary tree
$\T$ of  $n$ leaves that is embedded in the plane.
The tree $\T$ contains no nodes of degree 2 
and has the following additional properties: 
\begin{enumerate}[itemsep=-0.5ex, leftmargin=0.8cm]
	\item[$\bullet$]
	$m$ out of the $n$ leaves of $\T$ 
	have been \emph{marked}, and the remaining  $r=n{-}m$ leaves are \emph{unmarked}
	(see \cref{fig:combToriginal}).  
	\item[$\bullet$]
	Every marked leaf $\ell$ is associated with a \emph{neighborhood}, denoted $\nh(\ell)$, which is a 
	subtree of $\T$ rooted at $\ell$
	(see \cref{fig:combMarkedSetting}).  
	\item[$\bullet$]
	Every two consecutive marked leaves in the topological ordering of $\T$ have disjoint neighborhoods
	(see \cref{fig:combMarkedSetting}).  
\end{enumerate}

We call a binary tree $\T$ that follows these properties, a
\emph{marked tree}.
%
Given a marked tree $\T$, let $\T_u$ denote the 
\emph{unmarked tree} obtained
by deleting all the unmarked leaves of $\T$ 
and contracting the resulting degree-$2$
nodes, see \cref{fig:combTmarked}.
We apply to  $\T_u$  the following definition, which is extracted 
from the proof of \cref{thrm:Aggarwal} in~\cite{AGSS89}, see \cref{fig:combTtogether}.
 
\begin{definition} 
	\label{def:nodes}
	Let $T$ be a proper binary tree and let $T^*$ be the tree
	obtained from $T$ after deleting all its leaves.
	A node $u$ in $T^*$ is called:
	\begin{enumerate}[itemsep=-0.5ex,label=\alph*),leftmargin=0.8cm]
		\item 
		\emph{Leaf} or  \emph{$L$-node} if $deg(u)=1$ in
		$T^*$, i.e.,  
		$u$ neighbors two leaves in $T$.
		\item  
		\emph{Comb} or \emph{$C$-node} if $deg(u)=2$ in $T^*$, 
		i.e., $u$ neighbors one leaf in $T$. 
		\item  
		\emph{Junction} or \emph{$J$-node} if $deg(u)=3$ in
		$T^*$, i.e., 
		$u$ neighbors no leaves in $T$.
	\end{enumerate}
	A \emph{spine} is a maximal sequence of consecutive
	$C$-nodes, which is delimited by $J$- or $L$-nodes.
         Each spine has two \emph{sides} and marked
    leaves may lie in either side of a spine.
\end{definition}


\begin{figure}
	\centering
	\begin{subfigure}[t]{0.48\textwidth}
	\centering
	\includegraphics[width=.96\textwidth,page=4]{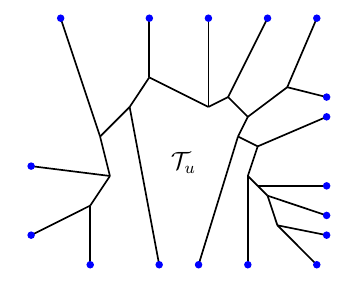}
	\caption{
		Tree $\T_u$. 
	}
	\label{fig:combTmarked}	
	\end{subfigure}
	\hfill
	\begin{subfigure}[t]{0.48\textwidth}
	\centering
	\includegraphics[width=.96\textwidth,page=5]{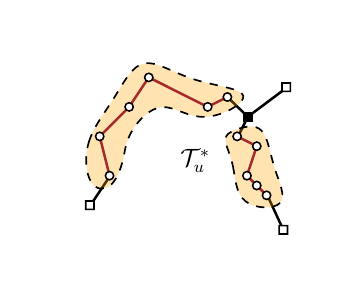}
	\caption{
		Tree $\T_u^*$:
		$L$-nodes are shown with ({\large$\square$}), $C$-nodes with ({\large$\circ$}), $J$-nodes with ({\large$\sqbullet$}), and spines are highlighted.
	}
	\label{fig:combTlabels}
	\end{subfigure}
	\caption{
		Illustration of \cref{def:nodes} applied to the tree $\T$ of \cref{fig:combToriginal}.
	}
    \label{fig:combTtogether}
\end{figure}

Let $\T_{u}^*$, be the tree obtained by applying \cref{def:nodes} to
the unmarked tree $\T_{u}$.
The nodes $\T_{u}^*$ are labeled as $L$-, $C$- and
$J$-nodes, see, e.g.,  \cref{fig:combTlabels}.
The labeling of nodes in $\T_{u}^*$ is then carried back to
their corresponding nodes in the original  marked tree $\T$ obtaining a \emph{marked tree $\T$ with
labels}, see~\cref{fig:combTcomponents}.
Some nodes in $\T$ remain unlabeled, see, e.g., node $u$  in
\cref{fig:combTcomponents}.

\begin{definition}
	\label{def:components}
	Given a marked tree $\T$ with labels we define the following two types of \emph{components}:
	\begin{enumerate}[label=\alph*),itemsep=-0.5ex,leftmargin=0.8cm]
		\item 
		\emph{$L$-component}: an $L$-node $\lambda$ defines an 
		$L$-component that consists of $\lambda$ union the two subtrees of $\T$
		that are incident to $\lambda$ and contain  no labeled node,
		see, e.g., $K_2$ in \cref{fig:combTcomponents}.
		The $L$-component contains exactly the two
		marked leaves that labeled $\lambda$.
		\item
		\emph{\group component}: 
		a group of five successive $C$-nodes 
		$c_i,\dots,c_{i+4}$ on a spine defines a
		\group component that consists of the path 
		$P_{c_i:c_{i+4}}$ from $c_i$ to  $c_{i+4}$
		(which may contain unlabeled nodes)
		union  the subtrees of $\T$, which are incident
		to the nodes of $P_{c_i:c_{i+4}}$ and contain no labeled node,
		see, e.g., $K_1$ in \cref{fig:combTcomponents}.
		Nodes $c_i$ and $c_{i+4}$ are referred to as the
		\emph{extreme nodes} of $K$.
		The  \group component contains exactly the five
		marked leaves, which labeled the five $C$-nodes. 
	\end{enumerate}
	Each spine is partitioned into consecutive  groups of \group
	components and at most four remaining 
	\emph{ungrouped $C$-nodes}. 
\end{definition}

\begin{figure}
	\centering
	\begin{minipage}[t]{0.48\textwidth}
		\centering
		\includegraphics[width=.96\textwidth,page=6]{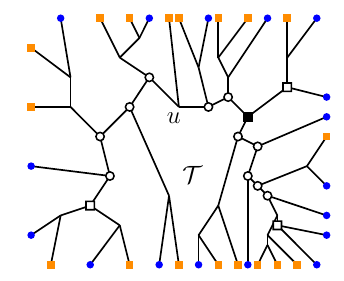}
		\caption{
		The marked tree $\T$ of \cref{fig:combToriginal} with labels ({\large$\circ$}, {\large$\square$}, {\large$\sqbullet$}).
		Node $u$ is not labeled.
		}
	\label{fig:combTlabelled}
	\end{minipage}
	\hfill
	\begin{minipage}[t]{0.48\textwidth}
		\centering
		\includegraphics[width=.96\textwidth,page=7]{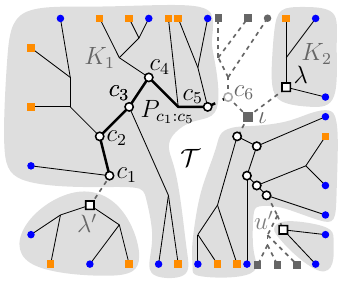}
		\caption{
		The components of $\T$ shown shaded. The dashed parts 
		do not belong to any component.
		}
	\label{fig:combTcomponents}
	\end{minipage}
\end{figure} 

\cref{fig:combTlabelled} and \cref{fig:combTcomponents} illustrates these definitions.
The tree $\T$ 
has three $L$-components and two \group components which 
are indicated shaded in  \cref{fig:combTcomponents}. 
The \group component $K_1$ contains the path $P_{c_1:c_5}$ from $c_1$
to $c_5$,  which is shown in thick black lines, and contains one unlabeled node.
Node $c_6$
is an ungrouped $C$-node.
\cref{fig:combTcomponents} also illustrates  a spine
consisting of the $C$-nodes $c_1,c_2,c_3,c_4,c_5,c_6$. 
The spine is
delimited by the $L$-node $\lambda'$ and the $J$-node~$\iota$;
it has five marked leaves from one side and one marked leaf from
the other.

\begin{observation}\label{obs:compDisjoint}
  The components of $\T$ are pairwise vertex disjoint.
  Every $L$-component contains exactly two marked leaves and every
  \group component contains exactly five marked leaves.
\end{observation}

Among the components of $\T$ there may be
subtrees of $\T$ consisting  of unlabeled nodes
and unmarked leaves that may be arbitrarily large. These subtrees hang off any unlabeled nodes
and ungrouped $C$-nodes. 
For example, in \cref{fig:combTcomponents}, node $u'$
is unlabeled and the gray dotted subtree incident to it
consists solely of unmarked leaves and unlabeled nodes that do not
belong to any component.
       
\section{Existence of leaves with pairwise disjoint
  neighborhoods}\label{sec:existence}

Aggarwal\etal~\cite{AGSS89} showed that for every eight ungrouped
$C$-nodes in $\T_u$ there exists at least one $L$-node.
Their argument 
holds for the marked tree $\T$ as well, which is described in the
following lemma for completeness.

\begin{lemma}
	\label{lem:oneten}  
	For every eight ungrouped $C$-nodes in $\T$ there exists at least one $L$-component.
\end{lemma}
      
\begin{proof}
	We count the $L$-nodes of $\T$ using the tree $\T_u^*$
	following the argument of \cite{AGSS89}.
	Let $k$ be the number of leaves in $\T_u^*$, which also equals
	the number of $L$-nodes in $\T$.
	Contracting all degree-2 vertices in $\T_u^*$  yields a binary
    tree $\T_{b}^*$, which has the same leaves as $\T_u^*$.
	Since $\T_{b}^*$ is an unrooted binary tree with $k$
	leaves, it has $2k-2$ nodes and  $2 k -3$ edges.
	Every edge in $\T_{b}^*$ corresponds to at most one spine in $\T_u^*$ and in every spine there are at most four ungrouped $C$-nodes. 
	Thus, 
	\[		
		|\text{ungrouped } C\text{-nodes}| \leq 4 |\text{spines}| \leq 4 \cdot (2 k-3) < 8 |L\text{-nodes}|,
	\]
	where $|\cdot|$ denotes cardinality.
	So, there exists at least one $L$-node for every eight
	ungrouped $C$-nodes, and an $L$-node corresponds to exactly
	one $L$-component.  
\end{proof}

The following lemmata establish that there exists a  constant
fraction of the marked leaves, which have pairwise disjoint neighborhoods.
The counting arguments follow those in~\cite{AGSS89} while they
are further enhanced to account for 
the unmarked leaves,
which are arbitrarily distributed among the marked leaves.
We say that the  neighborhood $\nh(\ell)$ of a marked leaf $\ell$ is
\emph{confined to a component} $K$ if it is a subtree of $K$.

\begin{figure}
	\centering
	\begin{subfigure}[t]{0.48\textwidth}
		\centering
		\includegraphics[width=\textwidth,page=2]{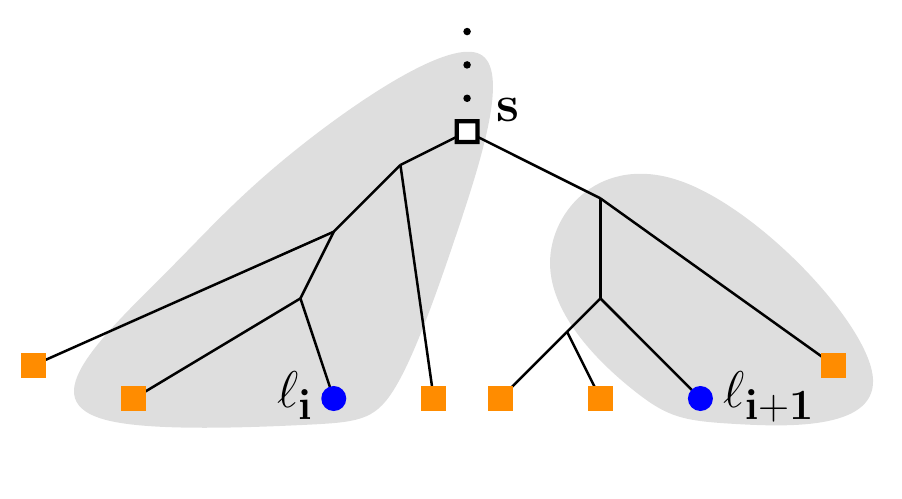}
		\caption{An $L$-component}
		\label{fig:combCasesLeaf}
	\end{subfigure}
	\hfill	
	\begin{subfigure}[t]{0.48\textwidth}
		\centering
		\includegraphics[width=\textwidth,page=3]{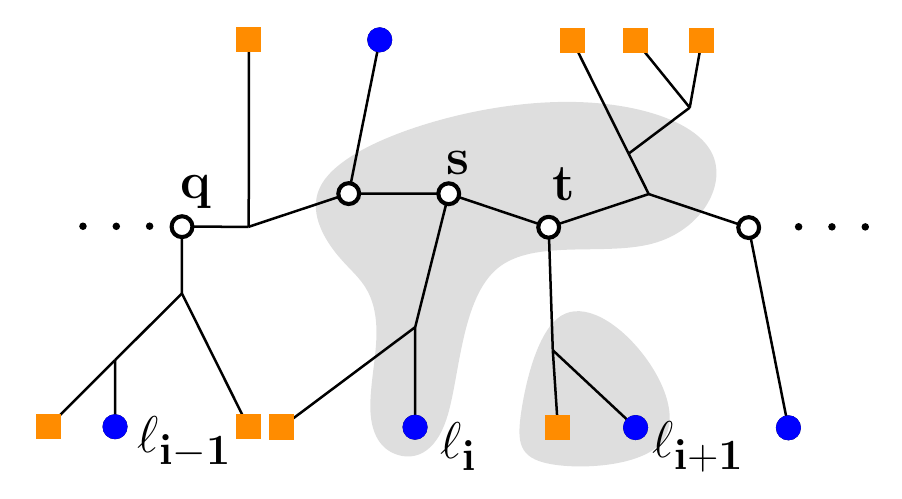}
		\caption{A \group component }
		\label{fig:combCasesGroup}
	\end{subfigure}	
	\caption{
		Marked leaves with their neighborhoods shaded.
		The neighborhood ${nh}(\ell_i)$ 
		is confined to the component in both cases.
	}
	\label{fig:combCases}
\end{figure} 

\begin{lemma}
	\label{lem:oneMarkedComp}
	In every component $K$, there exists a marked leaf $\ell \in
	K$ whose neighborhood is confined to $K$.
	This neighborhood may contain no $L$-node and no extreme
	$C$-node.
\end{lemma}

\begin{proof}
Let $K$ be an $L$-component and let $s$ be the
$L$-node that defines $K$.
Let $\ell_i$ and
$\ell_{i+1}$ be the two marked leaves of $K$.
Since  the neighborhoods $\nh(\ell_i)$ and $\nh(\ell_{i+1})$ are
disjoint,
at least one of  them 
cannot contain $s$.
This neighborhood is, thus, entirely contained in the relevant subtree
rooted at $s$, see \cref{fig:combCasesLeaf}, and contains no labeled node.

Let $K$ be a \group component. 
Since a \group component  has two sides, at least three out of the
five marked leaves of the component must lie on the
same side of $K$,
call them $\ell_{i-1},\ell_{i}$ and $\ell_{i+1}$.
Let $q,s$, and $t$ be their corresponding $C$-nodes,
i.e., the first $C$-nodes in $K$
reachable from
$\ell_{i-1},\ell_{i}$, and $\ell_{i+1}$, respectively,
see \cref{fig:combCasesGroup}.              
There are three cases.
If $t \in \nh(\ell_i)$, then $t \notin
\nh(\ell_{i+1})$ (since the two neighborhoods are
disjoint), and thus, $\nh(\ell_{i+1})$ is
confined to the subtree of $t$ that contains $\ell_{i+1}$. 
Similarly, if $q \in \nh(\ell_i)$, then 
$q \notin \nh(\ell_{i-1})$, so
$\nh(\ell_{i-1})$ is confined to the subtree of $q$ containing $\ell_{i-1}$.
If neither $q$ nor $t$ are in $nh(\ell_i)$,
then clearly $\nh(\ell_{i})$ is  confined to $K$.
In all cases the confined neighborhood cannot contain
neither $q$ nor $t$.
So, at least one of the five marked leaves must have a
neighborhood confined to $K$
and this neighborhood
cannot contain the extreme 
$C$-nodes in $K$.
\end{proof}

\begin{lemma}
	\label{lem:existence}
	Let $\T$ be a marked tree with $m$ marked leaves.
	At least $\frac{1}{10}m$ marked leaves 
	must
	have pairwise disjoint neighborhoods 
	such that 
	no tree edge 
	may have its endpoints in 
	two different neighborhoods.
\end{lemma}
\begin{proof}

	Every spine of $\T$ has up to four ungrouped $C$-nodes. 
	By \cref{lem:oneten}, there exists at least one $L$-component for every eight ungrouped $C$-nodes.
	By \cref{lem:oneMarkedComp}, every component of $\T$ 
	has at least one marked leaf whose neighborhood is confined to
	the component.
	So, overall, at least $\frac{1}{5}$ of the marked leaves from
	each \group component
	and at least $\frac{1}{10}$
	marked leaves of the remaining nodes,
	which label ungrouped $C$-nodes or $L$-nodes,
	have a confined neighborhood.
	The components are pairwise disjoint, so at least
	 $\frac{1}{10}$ marked leaves have pairwise disjoint neighborhoods.
	Furthermore, confined neighborhoods do not contain any
	$L$-node or extreme $C$-node, as shown in \cref{lem:oneMarkedComp}.
	Thus, no tree edge may have its endpoints in two
	different neighborhoods.
\end{proof}

We remark that the neighborhoods implied by
\cref{lem:existence} may not
contain any $L$-node nor any extreme $C$-node.
We also remark that these neighborhoods 
need not be of constant complexity
as their counterparts in
\cite{AGSS89} are. 
These neighborhoods may have complexity  $\Theta(r)$, where $r=n-m$ is
the number of unmarked leaves. 
Since  $r$ may be $\Theta(n)$, this poses a  challenge 
on how we can select these leaves 
efficiently.

\section{Selecting leaves with pairwise disjoint neighborhoods}
\label{sec:alg}
Given a  marked tree $\T$ with $m$ marked leaves, we have already
established the existence of  $\frac{1}{10}m$  marked leaves
that have pairwise disjoint neighborhoods.
In this section, we present an algorithm to select a fraction $p$ of
these leaves, i.e., $\frac{p}{10}m$ 
marked leaves with pairwise disjoint neighborhoods, in time
$O(\frac{1}{1-p}n)$, where $0<p<1$.

The main challenge over the algorithm of \cite{AGSS89} is
that the $r$ unmarked leaves are arbitrarily distributed among
the $m$  marked leaves,
and thus, the components of $\T$ and the neighborhoods of 
the marked leaves
may have complexity $\Theta(r)$.
If for each component we spend time
proportional to its size,  then 
the time complexity of the algorithm will be $\Theta(mr)$, 
i.e., $\Theta(n^2)$ if  $r,m\in \Theta(n)$.

To keep the complexity of the algorithm linear, 
we spend time up to a predefined number of steps in each
component depending
on the ratio  $c = \left\lceil \frac{r}{m} \right\rceil$
and the trade-off parameter $p\in (0,1)$. 
Our algorithm guarantees to find at least a fraction $p$ 
of the possible $\frac{1}{10}m$ marked leaves in time $O(\frac{1}{1-p}n)$. 
We first present a series of results 
necessary to establish the correctness of the approach 
and then  describe the algorithm.

Let $\ell_1, \dots, \ell_{m}$  be the marked leaves in $\T$
ordered in a counterclockwise topological ordering. 
Let the \emph{interval} $(\ell_i, \ell_{i+1})$ denote
the set of unmarked leaves between $\ell_i$ and
$\ell_{i+1}$ in the same order.
The \emph{interval tree} of  $(\ell_i,\ell_{i+1})$, denoted  $T_{(\ell_i,\ell_{i+1})}$,
is the minimal subtree of $\T$ 
that contains the marked leaves $\ell_i$ and $\ell_{i+1}$, including the unmarked
leaves in $(\ell_i,\ell_{i+1})$,
see \cref{fig:combIntervals2}.
We show the following \emph{pigeonhole lemma} involving unmarked
leaves and intervals. 

\begin{lemma}%
	\label{lem:pigeon}
	Suppose that  $r$ items (unmarked leaves) are distributed 
	in  $k\geq m$ containers (intervals), and let 
	$c = \left\lceil \frac{r}{m} \right\rceil$.
	For any natural number $x \leq r$, let $k_x$
        denote the
	number of containers that contain more than 
        $x$	items. 
	Then
	$k_x \leq \displaystyle \frac{cm}{x +1}$.
\end{lemma}		
\begin{proof}\let\qed\relax
	Each of the $k_x$ containers
	contains at least $x+1$ items. Thus,
	\begin{align}
		&
		k_x(x+1)  \leq r \  \Rightarrow  \ 
		k_x \ \leq \ \frac{r}{x +1}.  
		\label{eq1} \\
		&
		c \ = \ \displaystyle \left\lceil \frac{r}{m} \right\rceil \  \Rightarrow  \ 
		c \geq \ \frac{r}{m}  \ \Rightarrow \  
		r \ \leq \ cm
		\label{eq2}\\
		&
		(\ref{eq1}) \stackrel{(\ref{eq2})}{\Longrightarrow}\ k_x \ \leq \ \frac{cm}{x +1} 
		\label{eq3}
	\end{align}\qedhere
\end{proof}

\begin{figure}
	\centering
	\begin{subfigure}[t]{0.325\textwidth}
		\centering
		\includegraphics[width=\textwidth,page=1]{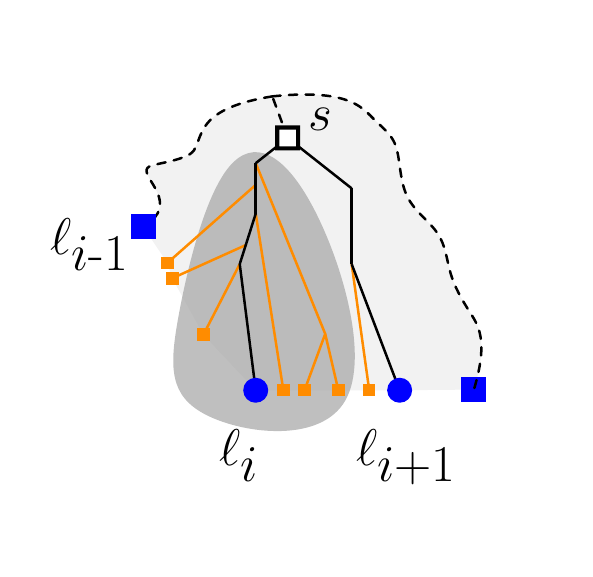}
		\caption{
			An $L$-component.
		}
		\hfill	
		\label{fig:combIntervals1}
	\end{subfigure}
	\begin{subfigure}[t]{0.325\textwidth}
		\centering
		\includegraphics[width=\textwidth,page=2]{combIntervals}
		\caption{\
			A \group component with $\ell_i=\ell_a$. 
			The interval tree {\color{DarkGreen}$T_{(\ell_j,\ell_{j+1})}$} is highlighted.
		}
		\label{fig:combIntervals2}
	\end{subfigure}	
	\hfill
	\begin{subfigure}[t]{0.325\textwidth}
		\centering
		\includegraphics[width=\textwidth,page=3]{combIntervals}
		\caption{
			A \group component with $\ell_i \notin \{\ell_a,\ell_b\}$.
			The interval trees {\color{DarkGreen}$T_{(\ell_{a},\ell_a^*)}$} and {\color{DarkGreen}$T_{(\ell_b,\ell_b^*)}$} are highlighted
		}
		\label{fig:combIntervals3}
	\end{subfigure}	
	\caption{
		Illustration of a component $K$ in different settings for the proof of \cref{lem:intervalTotree}. 
		The neighborhood $nh(\ell_i)$ is shaded gray.
		Marked leaves of $K$ are indicated with ({\color{Blue}$\bullet$})  and the other marked leaves with ({\color{Blue}$\sqbullet$}).
	}
	\label{fig:combIntervals}
\end{figure}

For a component $K$, let $\delta_K$ denote the maximum number of topologically consecutive
	unmarked leaves in $K$.	 The unmarked leaves counted in
        $\delta_K$ belong to some interval $(\ell_i,\ell_{i+1})$.

\begin{lemma}	\label{lem:intervalTotree}
	Let $K$ be a component of $\T$ 
	and let  $\ell_i$ be a marked leaf whose neighborhood
        $\nh(\ell_i)$ is confined to $K$. 
	\begin{enumerate}[label=\alph*),itemsep=-0.5ex, leftmargin=0.8cm]
	\item 
		If $K$ is an $L$-component, then  $\nh(\ell_i)$ has at most $4\delta_{K}$ nodes.  		    
	\item 
		If $K$ is a \group component, then $\nh(\ell_i)$ has at most $10\delta_{K}$ nodes.
	\end{enumerate}
\end{lemma}	
\begin{proof}
	Let $K$ be an $L$-component whose $L$-node is $s$,
	see \cref{fig:combIntervals1}.
	Since  $\nh(\ell_i)$ is confined to $K$ then  $s \notin
	\nh(\ell_i)$.
	Thus, $s$ disconnects $\nh(\ell_i)$ from the rest of $\T$, making
	$\nh(\ell_i)$ disjoint from any  interval tree, other than $T_{(\ell_{i-1},\ell_i)}$ and $T_{(\ell_i,\ell_{i+1})}$.
	Hence, $\nh(\ell_i)$ contains at most $2\delta_K+1$ leaves, and since 
	it is a proper binary tree, it can have at most $4\delta_K$ nodes in total. 

	Suppose $K$ is a \group component.
        Since $K$ contains exacly five marked leaves, there can be at most
        seven interval trees that may share a node with $K$.
	Let $a$ and $b$ be the two
	extreme $C$-nodes of $K$ and let
	$\ell_a$ and $\ell_b$ be their corresponding 
	marked leaves, which labeled $a$ and $b$ as $C$-nodes. 
	Let $\ell_a^*$ (resp. $\ell_b^*)$  be the neighboring marked
	leaf of $\ell_a$ (resp. $\ell_b)$ 
	in the topological ordering of the marked leaves, which does not belong to $K$.
	Refer to \cref{fig:combIntervals2} and \cref{fig:combIntervals3}.
	Neighborhood $\nh(\ell_i)$ is confined to $K$,
	thus,  $a,b \notin \nh(\ell_i)$.
	If $\ell_i= \ell_a$ (resp. $\ell_i= \ell_b$)
	the $C$-node $a$ (resp. $b$), disconnects $\nh(\ell_i)$ from the
	rest of $\T$.
	Thus, $\nh(\ell_i)$ 
	has a node 
	in common with 
	only two interval trees, 
	$T_{(\ell_{i-1},\ell_i)}$ and $T_{(\ell_i,\ell_{i+1})}$,
	see \cref{fig:combIntervals2}.
	If $\ell_i \notin \{\ell_a,\ell_b\}$, 
	then nodes $a$ and $b$ disconnect
	$\nh(\ell_i)$ from the rest of $\T$, thus,
        $\nh(\ell_i)$
	is disjoint from  both 
	$T_{(\ell_{a},\ell_a^*)}$ and $T_{(\ell_b,\ell_b^*)}$, see \cref{fig:combIntervals3}.
        Then $\nh(\ell_i)$ may have a node in common with at most
        five out of the seven interval trees that could be related to $K$.
        Concluding, $\nh(\ell_i)$ has at most $5\delta_K+1$ leaves, and since 
	it is a proper binary tree, it has at most $10\delta_K$ nodes overall. 	
\end{proof}

For each component $K$ we define a so-called \emph{representative
leaf} and at most two \emph{delimiting nodes}. 
These are used by our algorithm to identify a confined
neighborhood  within the component. 
    
\begin{definition}\label{def:checknodes}
	For a component $K$, we define its  \emph{representative
	leaf} and \emph{delimiting nodes} as follows:
	\begin{enumerate}[label=\alph*), itemsep=-0.5ex, leftmargin=0.5cm]
		\item 
		If $K$ is an $L$-component,
		there is one \emph{delimiting node}, which is
		its  $L$-node.
		The \emph{representative leaf} is
		the first marked leaf of $K$ in the topological
		ordering of leaves.
		In \cref{fig:combCasesLeaf}, $\ell_i$ is the representative leaf and $s$ is the  delimiting node.
		\item 
		If $K$ is a \group component, consider the side of $K$ containing at least three marked leaves.
		The \emph{representative leaf} is the
		second leaf among these three leaves in the topological ordering.
		The \emph{delimiting nodes} are the $C$-nodes
		defined by the other two leaves in the same side.
		In \cref{fig:combCasesGroup}, $\ell_i$ is the representative leaf and $q,t$ are the delimiting nodes.
	\end{enumerate}
\end{definition}

Our algorithm takes as input a marked tree $\T$ and a
parameter $p\in (0,1)$, and returns $\frac{p}{10}m$
marked leaves that
have pairwise disjoint neighborhoods.
A pseudocode description is given in \cref{algorithm}.
The algorithm iterates over all the components of $\T$, and 
selects at most one marked leaf for each component.

For each component $K$, the algorithm first
identifies its representative leaf and delimiting nodes (lines~6,13),
and then 
traverses the neighborhood of the
representative leaf performing a depth-first-search in the
component up to
a predefined number of steps (lines~7,14).
If,  while traversing the neighborhood, a delimiting node is detected
(lines~8,15,17), 
then a marked leaf is selected (lines~9,16,18), following the case
analysis of \cref{lem:oneMarkedComp}.
If the entire neighborhood is traversed within the allowed number
of steps 
without detecting a 
delimiting node (lines~10,19), then 
the representative leaf is selected (lines~11,20).
Otherwise, 
$K$ is abandoned and the algorithm proceeds to the next component.
\begin{algorithm}[t!]
	\SetAlgoVlined \SetAlgoNoEnd 
	\LinesNumbered
	\SetKwInOut{Input}{Input} \SetKwInOut{Output}{Output}
	\SetKwFunction{nh}{$nh$} 
	\SetKw{break}{break}
	\SetKwProg{Trace}{for}{ traverse $\nh(\ell_i)$}{end} 			
	\SetKwBlock{Whiledos}{while}{do}			
	\Input{A marked tree $\T$ with $n=r+m$ leaves and a parameter $p \in (0,1)$.}
	\Output{A set \textit{sol} of marked leaves.}
	\vspace*{1.5mm}
		Obtain the labeling of $\T$\;
		Partition $\T$ into components as indicated in \cref{def:components}\;
		$sol$ $\gets$ $\emptyset$; \tabb 
		$c \gets \left\lceil \displaystyle\frac{r}{m} \displaystyle\right\rceil$; \tabb
		$z \gets \left\lceil \frac{10c}{1-p}\right\rceil -1$\;		
		\For{\textup{each} component $K$ of $\T$}{
			\If{$K$ is an $L$-component}{
				$\ell_i \gets$ representative leaf; \tabb $s \gets$ delimiting node\;
				\Trace{at most $4z$ steps}{
					\If{$s$ is visited}{
						$sol$ $\gets$ $sol \ \cup \ \{\ell_{i+1}\}$; \tabb \break\;
					}
				}
				\If{$\nh(\ell_i)$ is traversed \textbf{and} $s$ is not visited}{
					$sol$ $\gets$ $sol \ \cup \ \{\ell_{i}\}$\;		
				}		        		        
			}    
			\ElseIf{$K$ is a \group component}{
				$\ell_i \gets$ representative leaf, \tabb $q,t \gets$ delimiting nodes\;
				\Trace{at most $10z$ steps}{
					\If{$q$ is visited}{
						$sol$ $\gets$ $sol \ \cup \ \{\ell_{i-1}\}$; \tabb \break\;
					}
					\If{$t$ is visited}{
						$sol$ $\gets$ $sol \ \cup \ \{\ell_{i+1}\}$; \tabb \break\;
					}					
				}
				\If{$\nh(\ell_i)$ is traversed \textbf{and} $q,t$ are not visited}{
					$sol$ $\gets$ $sol \ \cup \ \{\ell_{i}\}$\;		
				}		          			    
			}
		}			
	\Return{$sol$}\; 
	\caption{Selecting leaves with pairwise disjoint neighborhoods.}%
	\label[algo]{algorithm}	
\end{algorithm}
\begin{lemma}
	\cref{algorithm} returns at least $\frac{p}{10}m$ 
	marked leaves with pairwise disjoint neighborhoods such that no tree edge has its endpoints in two different neighborhoods.
	\label{lem:correctness}
\end{lemma}

\begin{figure}
	\centering
	\includegraphics[width=0.9\textwidth,page=4]{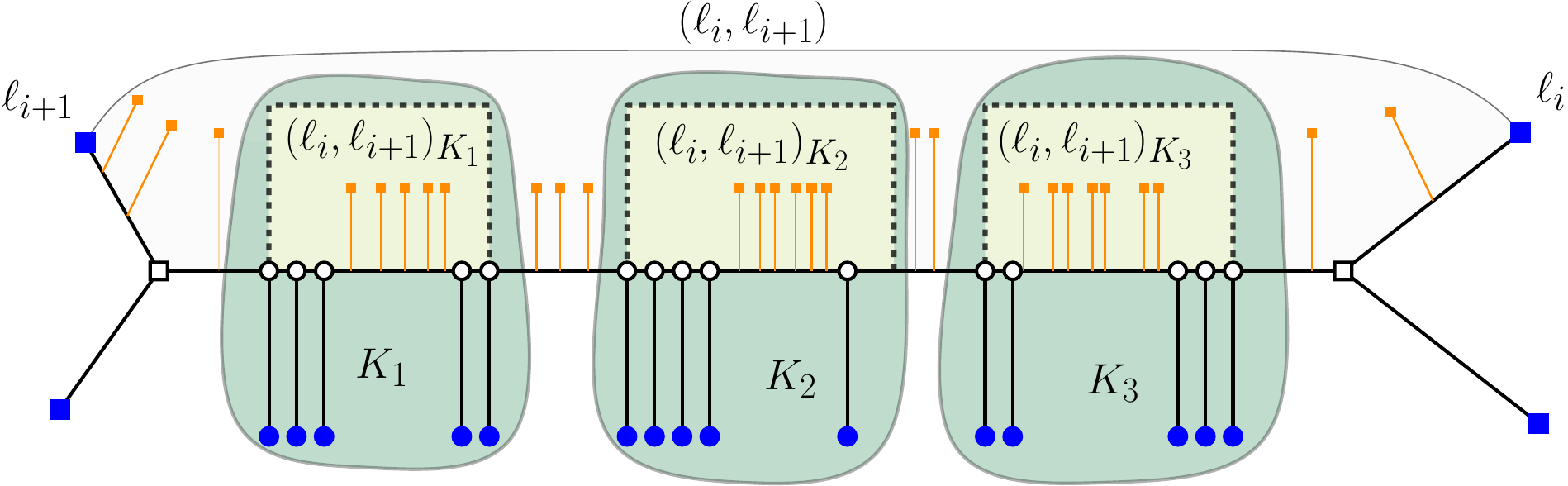}
	\caption{
		An interval $(\ell_i,\ell_{i+1})$ related to three components $K_1$,$K_2$ and $K_3$.
		Interval $(\ell_i,\ell_{i+1})$ is further subdivided into three intervals $(\ell_i,\ell_{i+1})_{K_1}$,$(\ell_i,\ell_{i+1})_{K_2}$ and $(\ell_i,\ell_{i+1})_{K_3}$.
	}
	\label{fig:combLargeInterval} 
\end{figure}    

\begin{proof}	
	Let $K$ be a component.
	The algorithm traverses the neighborhood 
	of the representative 
	leaf $\ell_i$ and takes a decision after 
	at most $4z$, or $10z$, steps. 
	In \cref{lem:intervalTotree}, we proved that 
	if $\nh(\ell_i)$ is confined, 
	$\nh(\ell_i)$ has at most $4\delta_K$, or $10\delta_K$, nodes.
	Hence, if $\delta_K \leq z$, 
	the algorithm will succeed to select a marked leaf from $K$, because
	either $\nh(\ell_i)$ is confined to $K$,
	and thus, the entire $\nh(\ell_i)$ is traversed 
	(lines 10-11,19-20), 
	or else a delimiting node gets visited,
	and thus, the corresponding marked leaf is selected
	(lines 8-9,15-18).
	In all cases, we follow the proof of \cref{lem:oneMarkedComp} and the neighborhood of the selected leaf is confined to $K$.
	Thus the selected leaf is among those counted in
	\cref{lem:existence}.

	If  on the other hand $\delta_k > z$, then the algorithm may fail to identify a marked
	leaf of $K$. We use the pigeonhole \cref{lem:pigeon} to  bound
    the number of these components.
    To this aim, we consider the set $I$ of all intervals
    induced by the marked leaves and the component of
    $\T$.
    For an interval  $(\ell_i,\ell_{i+1})$, which is not disjoint
    from $K$, let
	$(\ell_i,\ell_{i+1})_K := (\ell_i,\ell_{i+1}) \cap K$ denote its
	sub-interval of unmarked leaves that belong to$K$, see an example in \cref{fig:combLargeInterval}.
    Let  $I_z$ be the intervals in $I$ that contain more than $z$ unmarked leaves.
	Then the  algorithm may fail in at most  $|{I_z}|$  components.

    To bound  $|{I_z}|$, we use \cref{lem:pigeon} for   
 	$x = z = \left \lceil\frac{10c}{1-p} \right \rceil -1$.  Then,
	\begin{align}
		|I_z| \ \stackrel{(\ref{eq3})}{\leq} \ \frac{cm}{z+1} \ = \  
		\frac{cm}{\left \lceil\frac{10c}{1-p} \right \rceil -1+1} \ \leq \ 
		\displaystyle \frac{cm}{\frac{10c}{1-p}} \ = \
		\frac{1-p}{10}m 
	\label{eq4}
	\end{align}

	Thus, the algorithm may fail for at most $\frac{1-p}{10}m $ components.
	By \cref{lem:oneten}, there exist at least $\frac{1}{10}m$ components
	in $\T$,
	thus, the algorithm will succeed in selecting a marked leaf from
	at least 
	\begin{align}
		\frac{1}{10}m - |{I}_z| \ \stackrel{(\ref{eq4})}{\geq}  \ \frac{1}{10}m - \frac{1-p}{10}m  \ = \ m\frac{p}{10}
	\end{align} 
	components, concluding the proof.
\end{proof}

\begin{lemma}
	\cref{algorithm} has time complexity $O(\frac{1}{1-p}n)$.
	\label{lem:runtime}
\end{lemma}
\begin{proof}
	Labeling and partitioning the  tree $\T$ into components can
	be  done in $\Theta(n)$ time.
	Then, for each component the algorithm traverses a neighborhood
	performing at most $10z =\Theta(\frac{c}{1-p})$ steps.
	There are $\Theta(m)$ components,
	so we have $O(\frac{c}{1-p} \cdot m)$ time complexity.         
	Recall that $c = \left\lceil \frac{r}{m} \right\rceil$. 
	If $m =\Theta(n)$, then $c=\Theta(1)$, so $cm = \Theta(n)$.
	Else if $m =o(n)$, then $cm = \Theta(r) =\Theta(n)$.
	In all cases, the time complexity of the algorithm is  $O(\frac{1}{1-p}n)$.
\end{proof}

By combining \cref{lem:existence,lem:correctness,lem:runtime}
we establish (and re-state)
\cref{thrm:main}. 
\setcounter{theorem}{1}
\begin{theorem}
	Let $\T$ be a marked tree of $n$ total leaves and $m$ marked leaves.
	Then there exist at least $\frac{1}{10}m$ leaves in $\T$ with pairwise disjoint neighborhoods such that no tree edge has its endpoints in two different neighborhoods.        
	We can select at least a fraction $p$ of these $\frac{1}{10}m$ 
	marked leaves in time $O(\frac{1}{1-p}n)$, for any $p \in (0,1)$.
\label{th:gcl}
\end{theorem}

If the parameter $p \in (0,1)$ is a constant, then the algorithm returns 
a constant fraction of the marked leaves and 
the time complexity of the algorithm is $O(n)$.

	\bibliographystyle{abbrv}
	\bibliography{combLemma}

\end{document}